\documentclass[journal]{IEEEtran}

\ifCLASSINFOpdf
\else
\fi

\usepackage{amsmath,amsthm,stfloats,amsfonts,algorithm,algorithmic,graphicx,multirow,bm,cite,caption,array,setspace}

\allowdisplaybreaks

\begin{document}
%
\title{Distributed Pricing-Based User Association for
	Downlink Heterogeneous Cellular Networks}

\author{Kaiming~Shen
        and Wei~Yu,~\IEEEmembership{Fellow,~IEEE} 
\thanks{Manuscript submitted to JSAC Special Issue on 5G Wireless Communication Systems on Dec. 1, 2013, revised on Apr. 5, 2014, accepted on May 3, 2014.
The materials in this paper have been presented in part in IEEE
International Conference on Acoustic, Speech, and Signal Processing (ICASSP),
May 2013, Vancouver, Canada \cite{Shen}. The authors are with
the Edward S. Rogers Sr.  Department of Electrical and Computer
Engineering, University of Toronto, Toronto, ON M5S 3G4, Canada
(e-mails: \{kshen,weiyu\}@comm.utoronto.ca).

This work is supported in part by BLiNQ Networks, in part by
National Science and Engineering Research Council (NSERC) of
Canada, and in part by Ontario Centre of Excellence (OCE).}%
}

\markboth{}
{Shen, Yu: Distributed Pricing-Based User Association for
Downlink Heterogeneous Cellular Networks}

\maketitle


\begin{abstract}
This paper considers the optimization of the user and base-station (BS) association
in a wireless downlink heterogeneous cellular network under
the proportional fairness criterion. We first consider the case where each
BS has a single antenna and transmits at fixed power, and propose a
distributed price update strategy for a pricing-based user association
scheme, in which the users are assigned to the BS based on the value
of a utility function minus a price. The proposed price update
algorithm is based on a coordinate descent method for solving the dual
of the network utility maximization problem, and it has a rigorous
performance guarantee. The main advantage of the proposed algorithm as
compared to the existing subgradient method for price update is that the proposed algorithm is
independent of parameter choices and can be implemented asynchronously.
Further, this paper considers the joint user association and BS power
control problem, and proposes an iterative dual coordinate descent and the
power optimization algorithm that significantly outperforms existing
approaches. Finally, this paper considers the joint user association and BS
beamforming problem for the case where the BSs are equipped with multiple
antennas and spatially multiplex multiple users.  We incorporate dual
coordinate descent with the weighted minimum mean-squared error (WMMSE)
algorithm, and show that it achieves nearly the same performance as a
computationally more complex benchmark algorithm (which applies the
WMMSE algorithm on the entire network for BS association), while
avoiding excessive BS handover.
\end{abstract}


\IEEEpeerreviewmaketitle

\section{Introduction}

\IEEEPARstart{M}{odern} wireless networks are designed based on the
cellular architecture in which multiple user terminals are associated
with the base-stations (BSs) to form cells. The cellular concept has
further evolved to include heterogeneous networks (HetNets)
where the BSs can transmit with widely different powers at disparate
locations, and consequently the cells can vary considerably in size.
An essential feature of HetNet is that it allows the off-loading of
traffic from the macro BSs to pico or femto BSs.  By splitting the
conventional macro cellular structure into small cells (i.e.,
femto/pico cells), the
HetNets allow for more aggressive reuse of frequencies as well as
improved coverage and higher overall throughput for the entire
network.

A main challenge in the deployment of HetNet is the appropriate
setting of the transmit power levels at different tiers of
macro/pico/femto BSs and the association of users to the different BSs
(or equivalently the determination of coverage area for each cell). The cell
association problem is further compounded when multiple antennas are
deployed at the BSs with multiple users spatially multiplexed using
multiple-input and multiple-output (MIMO) beamforming techniques.
Conventionally, the downlink cell coverage areas are
determined according to the signal-to-interference-plus-noise ratio
(SINR).  Each user terminal simply associates with the BS from which
the received SINR is the highest---herein referred to as the max-SINR
rule. A key problem with the max-SINR BS association is that it does
not account for the varying data traffic pattern in the network, hence
it can lead to poor load balancing.  Load balancing is essential for
wireless networks with small cells, because femto/pico BSs are often
deployed to alleviate traffic ``hot-spots" with higher-than-average
user density.

This paper addresses the downlink user association problem for HetNets
from an optimization perspective under the proportional fairness
criterion. We follow a pricing-based strategy in which the users
are associated with the BS according to the value of a utility
minus a price---a strategy first adopted in \cite{andrews}, where
a price update method based on the subgradient
algorithm is proposed.  The main novelty of this paper is that we
advocate an alternative price update method based on a coordinate
descent approach on the dual of the network utility maximization
problem. The proposed algorithm has the advantage that it is free
of parameter choices and that it can be implemented asynchronously
across the BSs. This paper further proposes joint optimization of
BS association with downlink power control and with beamforming.
We show that the proposed
pricing-based distributed user association can significantly improve
the conventional max-SINR association.
Throughout the paper, we use
the terms {\it BS association} and {\it user association}
interchangeably--- the former emphasizes a user perspective, while the latter a BS perspective.

\subsection{Related Work}

The BS association problem has been considered extensively in the
literature. While the early works in this area \cite{yates, hanly,
farrokh, shroff, goodman} mostly deal with the code-division multiple access
(CDMA) system, they already reveal that the joint optimization of BS
association and transmit power can significantly improve the overall
network performance. These earlier works, as well as some of the more
recent ones \cite{lev,qian,silva,jtwang,harri,duy}, tend to focus on
the power-based optimization objectives, e.g., minimizing
the total transmit power under a predefined set of minimum SINR
constraints at the user terminals.
While the power minimization formulation may be appropriate for
networks with fixed rate and fixed quality-of-service (QoS)
requirement,
modern wireless networks often maximize the objective of
the overall throughput, or more generally, a network utility function
across all users in the network. In this realm, \cite{kuang, chitti,
corroy, vu} consider the maximization of the sum rate across the
network, while \cite{hong, wmmse_luo} consider the weighted sum rate
objective for the BS association problem.  More general network
utility maximization formulation is considered in \cite{andrews,
C3, guvenc, bu}, which use a proportional fairness objective
function, while \cite{corroy} considers max-min fairness in addition.

This paper considers the network utility maximization problem under
the proportional fairness, i.e., log-utility objective for the downlink
of a wireless cellular network.
Because the BS association problem is inherently a discrete
optimization problem involving the assignment of users to BSs,
finding the optimal solution for such a problem is nontrivial.
While conventional BS association simply uses the max-SINR rule, it is
also clear from a network utility maximization or load-balancing
perspective that max-SINR is far from being adequate. In this
direction, \cite{guvenc} proposes an intuitive idea of expanding the
coverage area of small cells by adding constant bias terms to the SINR
values, so as to balance the load among different cells (although
\cite{guvenc} does not analyze what the optimal bias terms should be).
Other common heuristics proposed in the literature include
that in \cite{bu}, \cite{C3}, \cite{madan}, which optimize BS
association through the greedy method, and \cite{lee}, which
randomly assigns each user terminal to the BS with the probability
proportional to the estimated throughput, and \cite{kuang, chitti, kim},
which devise their respective methods based on the relaxation
heuristic. In addition to the network utility maximization
formulation, \cite{jiang} addresses the BS assignment problem from a
game theory perspective (as the assignment problem can be thought of
as a game among the BSs), where the Nash equilibrium of the game is
found.

The BS association algorithm proposed in this paper is most closely
related to \cite{andrews}, where under fixed transmit powers, a dual
pricing method based on the subgradient update is proposed. This paper adopts this pricing approach, but makes
further progress in identifying an alternative price update
method. Other related works on BS association assuming fixed
transmit powers include
\cite{corroy}, which considers a simple model consisting of only a
single pair of macro and pico BSs, and \cite{hong}, which considers a
special situation where user terminals may not report their channel
state information (CSI) truthfully out of selfish motivation.

For the purpose of load balancing and interference management,
it has also been well recognized in the existing literature that BS
association and transmit power levels need to be optimized jointly.
From this joint optimization perspective, an intuitive but heuristic
idea is to optimize BS association and power levels in an iterative
fashion, as suggested in \cite{chitti,vu,madan}.  The approach of
\cite{kim} addresses the joint optimization problem using duality
theory, but only for a relaxed version of the problem with the
discrete constraints eliminated. In general, BS association and power
optimization for weighted rate-sum maximization are both
challenging problems, but there are some special cases where the
globally optimal solution to the joint optimization problem can be
found. For example, in \cite{luo} the optimal settings of BS
association and power levels that maximize the sum throughput are
obtained under certain restricted conditions for the case
where the number of user terminals and the number of BSs are the same.
Instead of searching for globally optimal solutions, this paper treats
the joint BS association and power optimization problem from an
iterative optimization perspective.  Our main
contribution here is some key observations on the role of
pricing-based BS association in this heuristic approach.

For multi-cell networks with multiple antennas at the BSs, this
paper also considers the joint optimization of BS association and
beamforming for the scenario where multiple users can be spatially
multiplexed within each cell. In this domain, \cite{harri, lu} provide
algorithms for such a joint optimization problem, but only under
the power minimization objective. In \cite{kuang}, BS association,
transmit power and beamforming vectors are optimized through
coordinate descent. Note that the beamforming problem by itself
(assuming fixed BS association) is well studied in the literature
(e.g., \cite{farrokhi_1, wiesel, hayssam, wmmse, christensen}).  In
this regard, the WMMSE algorithm \cite{wmmse} is of particular
interest, because it can handle weighted rate sum maximization, hence
the proportional fairness objective. A recent work \cite{wmmse_luo}
proposes a modification of the WMMSE algorithm that is capable of optimizing
BS association and the beamforming vectors jointly. The WMMSE
algorithm of \cite{wmmse_luo} is, however, computationally complex; further it induces
excessive BS handover. One of the contributions of this paper is that
the pricing-based BS association can be incorporated with
WMMSE beamforming design to significantly reduce the computational
complexity of joint BS association and beamforming method of \cite{wmmse_luo}, while achieving
nearly the same performance and avoiding excessive handover.

\subsection{Main Contributions}
This paper considers the optimal joint BS association with power
control and with beamforming for the downlink HetNets under the
proportional fairness objective. The main contributions of this paper
are as follows:

\subsubsection{BS Association}
For a single-input and single-output (SISO) network with fixed
transmit powers and with flat-fading channels, this paper proposes a
distributed pricing-based user association scheme with a price update
method based on coordinate descent in the dual domain. The proposed
price update algorithm has faster convergence than the conventional
subgradient method \cite{andrews}. It is a fundamental building block
for subsequent generalizations to the frequency-selective case and to
the cases with power control and MIMO beamforming. Moreover, we
provide a duality-gap based analysis to bound the performance
error of the proposed algorithm.

\subsubsection{Joint BS Association and Power Control}
This paper proposes an iterative optimization approach for the joint
BS association and power control problem. We make a key observation
that the choice of BS association method is crucial in joint
optimization. In particular, when used in conjunction with power control,
the conventional max-SINR association
tends to exacerbate load imbalance, while the proposed pricing-based
association alleviates load imbalance.
To quantify the performance of the proposed iterative approach, we
devise a benchmark algorithm based on dual optimization and by solving
the nonconvex power optimization problem from multiple starting points. We show
that the proposed iterative approach provides comparable performance, while
being much less computationally complex.


\subsubsection{Joint BS Association and Beamforming}
When BSs are equipped with multiple antennas and have the ability to
spatially multiplex multiple users within each cell, this paper shows
that the optimization of BS association and beamforming can be
decoupled without significantly affecting the overall performance.
This allows us to propose a two-stage method combining the joint BS
association and power control algorithm as the first stage followed by
a per-cell WMMSE step in the second stage. The proposed approach is
significantly less complex than the use of WMMSE algorithm for BS
association over the entire network \cite{wmmse_luo}, while at the
same time avoiding excessive BS handover.

\subsection{Organization of This Paper}

The rest of this paper is organized as follows. Section II introduces
the problem formulation for the BS association problem for a SISO
network.  Section III analyzes a pricing based BS association approach
under fixed power. The algorithm proposed in Section III is a key
component in subsequent developments.  Section IV considers the
joint BS association and power control problem.  Section V addresses
the joint BS association and beamforming problem for a MIMO network.
Performance evaluations are provided in Section VI. Conclusions are
drawn in Section VII.


\section{BS Association Problem for SISO Networks}

Consider a downlink cellular network consisting of $L$ BSs with fixed
transmit power levels (which may differ from BSs to BSs), and $K$
active user terminals across the geographic area covered by the
network. Both the BSs and the user terminals are equipped with a
single antenna each. Let $i$ be the index of user terminals,
$i\in\{1,2,\ldots,K\}$, and let $j$ be the index of BSs,
$j\in\{1,2,\ldots,L\}$. Let the total bandwidth of the system be $W$,
which is shared by all BSs (i.e., frequency reuse factor is one).
To simplify the problem, we assume flat-fading channels and
frequency-flat PSD levels at the BSs, thus the SINR values are
constants across the frequencies. Let $h_{ij}\in\mathbb{C}$ be the
channel between user $i$ and BS $j$, and let $p_{j}$ be the transmit
PSD level at BS $j$. If the user $i$ is to be associated with the BS
$j$, its SINR value is then
\begin{equation}
\label{sinr}
\text{SINR}_{ij}\left(\textbf{p}\right) = \frac{|h_{ij}|^{2}p_j}{\sum_{j'\neq
j}|h_{ij'}|^{2}p_{j'}+\sigma_i^2},
\end{equation}
where $\textbf{p}=\left[p_1,\cdots,p_L\right]^T$; $\sigma_i^2$ is the PSD
of the background additive white
Gaussian noise (AWGN). This paper assumes that each user is associated
with one BS at a time.

This paper adopts a proportionally fair network utility optimization
framework of maximizing the sum log-utility across all the users in
the entire network. A key step in the problem formulation is an
observation made in \cite{andrews}, where it is shown that for a
given set of users associated with a BS, round-robin among these users
is the proportionally fair schedule (assuming constant and
flat-fading channel and flat transmit PSD).  Hence, if a total of
$k_j$ users are associated with BS $j$, in order to maximize the
proportional fairness objective, each of them should be allocated
$1/k_j$ of the total time/frequency resource. In this case, if a user
$i$ is associated with BS $j$, its rate is given by
\begin{equation}
R_{ij}\left(\textbf{p},k_j\right) = \frac{W}{k_{j}}
\log\left(1+\frac{\text{SINR}_{ij}}{\Gamma}\right),
\label{rate_BS}
\end{equation}
where $\Gamma$ is the SNR gap determined by practical coding and
modulation schemes used. 

Let $x_{ij}$ be a binary variable (1 or 0) denoting whether or not
user $i$ is associated with BS $j$. The BS association problem is that
of jointly determining $x_{ij}$ and the transmit powers $p_j$ at each
BS to maximize the overall network utility, which can be written as:
\begin{subequations}
\label{power_problem}
\begin{eqnarray}
\mathop\text{maximize}\limits_{\mathbf{X},\mathbf{p},\mathbf{k}} & &
\sum_{i,j}x_{ij}\log\left(R_{ij}(\mathbf{p},k_j)\right)
	\label{obj_Power} \\
\text{subject to}&&0 \leq p_j \leq \overline{p}_j, \; \forall j
        \label{power_constraint_Power} \\
&&\sum_{j} x_{ij} = 1, \; \forall i
	\label{pp_1} \\
&&\sum_{i} x_{ij}=k_j, \; \forall j
	\label{compute_k} \\
&&\sum_{j} k_{j} = K
	\label{pp_2} \\
&&x_{ij}\in\left\{ 0,1 \right\}, \; \forall i,\forall j
  \label{x_constraint_Power}
\end{eqnarray}
\end{subequations}
where $\mathbf{X}=[x_{ij}]$; $\mathbf{k}=[k_1,\cdots,k_L]^T$;
$\overline{p}_j$ is the PSD constraint of BS $j$. Constraint
(\ref{pp_1}) ensures that each user can only associate with one BS, and
constraint (\ref{pp_2}) states that all users in the network are
served. Note that although $k_j$ is completely determined by $x_{ij}$,
it is convenient to keep $k_j$ as an optimization variable in
subsequent analysis.

\section{BS Association in SISO Networks Under Fixed Power}

The joint BS association and power control problem (\ref{power_problem})
is a mixed discrete optimization (over the BS association) and
nonconvex optimization problem (over the powers), for which finding
its global optimum is expected to be very challenging. In this section,
we focus on a simplified problem setting with transmit power spectral
density (PSD) levels fixed a \emph{priori}. The joint optimization
problem with power control is treated in the subsequent section.

\subsection{Problem Formulation}

When $\textbf{p}$ is fixed in (\ref{power_problem}), all SINR values
are predefined by (\ref{sinr}). We introduce parameter \begin{equation}
a_{ij} =
\log\left(W\log\left(1+\frac{\text{SINR}_{ij}}{\Gamma}\right)\right).
\label{parameter_a}
\end{equation}
Substituting $a_{ij}$ back into (\ref{power_problem}), we simplify
the BS association problem under the fixed powers as
\begin{subequations}
\label{bsa_problem}
\begin{eqnarray}
\mathop\text{maximize}\limits_{\mathbf{X},\mathbf{k}} & &
\sum_{i,j}a_{ij}x_{ij} - \sum_{j}k_{j}\log\left(k_{j}\right)
	\label{obj2_BS} \\
\text{subject to}&&\sum_{j} x_{ij} = 1, \; \forall i
	\label{num_association_BS} \\
&&\sum_{i} x_{ij}=k_j, \; \forall j
    \label{num_user_BS}\\
&&\sum_{j} k_{j} = K
    \label{sum_num_user_BS}\\
&&x_{ij}\in\left\{ 0,1 \right\}, \; \forall i, \forall j
\end{eqnarray}
\end{subequations}
This rest of this section presents a pricing approach, together with a
novel price update method, for solving the above problem.

\subsection{Lagrangian Dual Analysis}

The problem formulation (\ref{bsa_problem}) is first proposed in
\cite{andrews}, where it is shown that a dual analysis can yield
considerable insight. An important idea is that the dual variables can
be interpreted as the BS-specific prices, which give rise to the dual
pricing approaches for BS association.

Introduce dual variables $\bm{\mu}=[\mu_1,\cdots,\mu_L]^T$ for
constraint (\ref{num_user_BS}), and $\nu$ for constraint
(\ref{sum_num_user_BS}). The Lagrangian function with respect to these
two constraints is
\begin{multline}
\label{Lagrangian}
L(\mathbf{X},\mathbf{k},\bm{\mu},\nu) = \sum_{i,j}a_{ij}x_{ij} - \sum_{j}k_j\log(k_j)\\
- \sum_j{\mu_j\left(\sum_{i} x_{ij}-k_j \right)} - \nu\left(\sum_{j} k_j-K\right).
\end{multline}
The dual function $g(\cdot)$ can then be written as
\begin{equation}
g(\bm{\mu},\nu)  = \left\{
\begin{array}{cl}
\mathop\text{maximize}\limits_{\mathbf{X},\mathbf{k}} & L(\mathbf{X},\mathbf{k},\bm{\mu},\nu) \\
{\rm s.t.\ } & \sum_{j} x_{ij}=1, \; i = 1, \ldots, K \\
&x_{ij} \in \{0,1\},\; \forall i,\forall j
\end{array}
\right.
\label{dual_not_enclosed}
\end{equation}
The maximization of the Lagrangian has the following explicit analytic solution:
\begin{equation}
x^*_{ij}=
\begin{cases}
1, \ \text{if } j=j^{(i)}\\
0, \ \text{if } j\neq j^{(i)}
\end{cases}
\text{where } j^{(i)}=\arg\mathop\text{max}\limits_{j'} \left(a_{ij'}-\mu_{j'}\right)
\label{x*}
\end{equation}
and
\begin{equation}
k_j^*=e^{\mu_j-\nu-1}.
\label{K*}
\end{equation}
Note that if $j^{(i)}$ in (\ref{x*}) is not unique, $x_{ij}$ can be
assigned value 1 for any of the BSs with maximum $(a_{ij}-\mu_j)$
without affecting the value of dual function.

The solution of $x_{ij}$ in (\ref{x*}) is quite intuitive. The dual
variable $\mu_j$ is the price at BS $j$, while $a_{ij}$ is the utility
of the user $i$ if it associates with BS $j$. Each user maximizes its
utility $a_{ij}$ minus the price among all possible BSs, while the BSs
choose their prices to balance their loads.

This pricing interpretation has already been given in \cite{andrews},
which also proposes a subgradient algorithm for updating the prices. The
present paper carries this idea one step further by observing that we
can explicitly write down the Lagrangian dual optimization problem of
(\ref{bsa_problem}). This additional observation gives rise to a better
price update method.

Substituting (\ref{x*}) and (\ref{K*}) back into
(\ref{dual_not_enclosed}), we obtain the dual objective in closed-form as:
\begin{equation}
g(\bm{\mu},\nu)=\sum_i\mathop\text{max}\limits_{j}
\left(a_{ij}-\mu_j\right)+\sum_j 
e^{\mu_j-\nu-1}
+\nu K.
\label{dual_enclosed}
\end{equation}
The Lagrangian dual problem of (\ref{bsa_problem}) is now the
minimization of $g(\cdot)$ over $\bm{\mu}$ and $\nu$:
\begin{equation}
\label{prob:dual}
\mathop\text{minimize }\limits_{\bm{\mu},\nu} g(\bm{\mu},\nu)
\end{equation}
The Lagrangian duality theory in optimization states that the updating
of the prices can be done via the minimization of $g(\bm{\mu},\nu)$,
e.g., using the subgradient algorithm \cite{andrews}.  One of the main
contributions of this paper is that by taking advantage of the
particular form of $g(\bm{\mu},\nu)$, the price update can
alternatively be done using a coordinate descent approach in the dual
domain. In the subsequent sections, we first review the subgradient
method, then present the new coordinate descent method.

After the dual solution is obtained for (\ref{prob:dual}), we need to
recover the primal variable $x_{ij}$ from the dual solution. This can
be done through (\ref{x*}), but there is the possibility that a user
has more than one BS with the same maximal value for $(a_{ij}-\mu_j)$.
Such ties can be resolved using heuristics. In general,
we would like to keep $k_j$ as close to $e^{\mu_j-\nu-1}$ as possible.
In our simulation experience, only a very small number of users are
typically involved in ties, so tie-breaking via exhaustive search is
feasible.

It should be noted that because the original optimization problem
(\ref{bsa_problem}) is discrete in nature, solving the dual is not the
same as solving the original primal problem---a positive duality gap
can exist.  Nevertheless, the dual optimum solution often leads to good
primal solutions.

\subsection{Subgradient Method}

To solve the dual optimization problem (\ref{prob:dual}), we
observe first that if $\bm{\mu}$ is fixed, then $g(\cdot)$ is a
differentiable convex function of $\nu$, so the optimal $\nu$ can be
found as\footnote{Strictly speaking, the algorithm described in this
section is a combination of alternating minimization between
$\bm{\mu}$ and $\nu$, and subgradient method on $\bm{\mu}$.
A full subgradient implementation would involve a subgradient update on $\nu$ as well, i.e.,
$
\nu^{(t+1)} = \nu^{(t)} - \alpha^{(t)}
\left( K - \sum_j e^{\mu_j^{(t)}-\nu^{(t)}-1} \right).
$
}
\begin{equation}
\nu^{(t+1)} = \log \frac{\sum_j e^{\mu_j^{(t)}-1}}{K},
\label{nu*}
\end{equation}
where the time index $t$ is included here to indicate that $\bm{\mu}$
and $\nu$ need to be updated iteratively in a sequential
order.
However, $g(\cdot)$ is not a differentiable function of $\mu_j$,
so instead of taking its derivative with respect to $\mu_j$, the
subgradient method updates $\mu_j$'s in each step according to
\begin{equation}
\mu_j^{(t+1)} = \mu_j^{(t)} - \alpha^{(t)} \left( e^{\mu_j^{(t)}
-\nu^{(t)}-1}-\sum_{i}x_{ij}^{(t)} \right), \; j = 1, \ldots, L
\label{subgradient}
\end{equation}
where $\alpha^{(t)}$ is the step size and $x_{ij}^{(t)}$ is determined
by $\mu_j^{(t)}$ according to (\ref{x*}).  The use of subgradient method
for price update in the BS association problem is first proposed in
\cite{andrews}.

Because the dual problem is always convex, the subgradient
method is guaranteed to converge to the globally optimal solution to the
dual problem (\ref{prob:dual}).  However,
the convergence speed of the subgradient method
depends heavily on the choice of step size $\alpha^{(t)}$. Possible
choices of $\alpha^{(t)}$ include constant step size (but the constant
is difficult to choose) or diminishing step sizes (which guarantee
convergence but can be quite slow in practice). As a baseline for
comparison, this paper adopts the self-adaptive scheme of
\cite{bertsekas_convex} as suggested in \cite{andrews}. We refer the
detailed algorithm description to \cite{bertsekas_convex}, and only
mention that the scheme involves quite a few parameters, namely
$\gamma_t$, $\rho\geq 1$, $\beta<1$, as well as $\delta_1$ and
$\delta$. 
Still, the convergence speed is still very much parameter dependent,
as seen in the simulation section later in this paper.

We remark that because all the $\mu_j$'s need to be updated at the
same time using the same step size (in order to ensure convergence),
the distributed implementation of the subgradient method requires
synchronized price updates across the BSs. This is a significant
drawback, as synchronization is not necessarily easy to achieve.  The
main advantage of the dual coordinate descent method proposed in the
next section is that it is free of parameter choices and it does not
require synchronization.

\subsection{Dual Coordinate Descent (DCD) Method}

The main contribution of this paper is a coordinate descent
\cite{bertsekas_coordinate}
approach in the dual domain for solving (\ref{prob:dual}).
The key idea is to recognize that
the dual function is expressed in a closed form in
(\ref{dual_enclosed}). First, fixing all the $\mu_j$'s, we see that optimal
$\nu$ can be updated by
(\ref{nu*}). Next, fixing $\nu$ and all $\mu_j$'s except one of them, we see
that $g(\cdot)$ is in fact the sum of a continuous piece-wise linear
function and an exponential function. So we can take its left or right
derivatives and choose $\mu_j$ to be such that the left derivative at
$\mu_j$ is less than or equal to zero, and the right derivative is
greater than or equal to zero. Mathematically, define two functions
$f_1(\cdot)$ and $f_2(\cdot)$ as:
\begin{equation}
f_1(\mu_j)=|\mathcal{U}_j|,
\label{f_1}
\end{equation}
where
$\mathcal{U}_j=\left\{i\left| a_{ij}-\mu_j=\max_{j'} (a_{ij'}-\mu_{j'}) \right. \right\}$,
and
\begin{equation}
f_2(\mu_j)=e^{\mu_j-\nu-1}.
\label{f_2}
\end{equation}
It is easy to see that the left partial derivative of $g(\cdot)$ with
respect to $\mu_j$ is exactly $f_2(\mu_j)-f_1(\mu_j)$. Hence, fixing
all other dual variables, the $\mu_j$ that minimizes $g(\cdot)$ is just
\begin{equation}
\mu_j^{(t+1)} = \sup \left\{\mu_j \left| f^{(t)}_2(\mu_j) -
f^{(t)}_1(\mu_j) \leq 0 \right\}
\label{mu_update}
\right. .
\end{equation}
This leads to the DCD method described in Algorithm 1.

\begin{figure}[t]
\begin{minipage}[b]{0.44\linewidth}
  \centering
	  \centerline{\includegraphics[width=4.8cm]{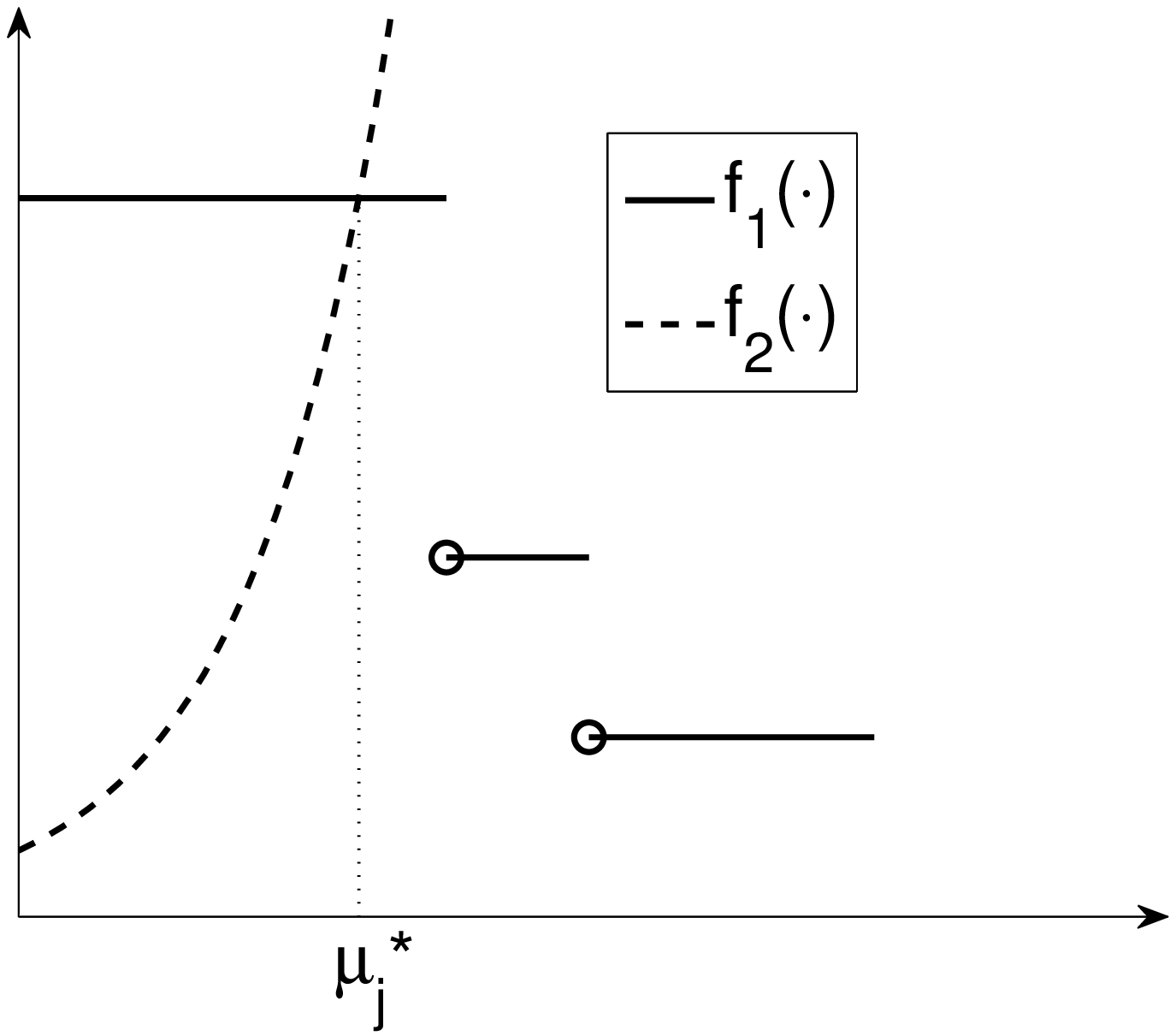}}
  \centerline{(a) $f_1$ and $f_2$ intersect}\medskip
\end{minipage}
\hfill
\begin{minipage}[b]{0.48\linewidth}
  \centering
	  \centerline{\includegraphics[width=4.8cm]{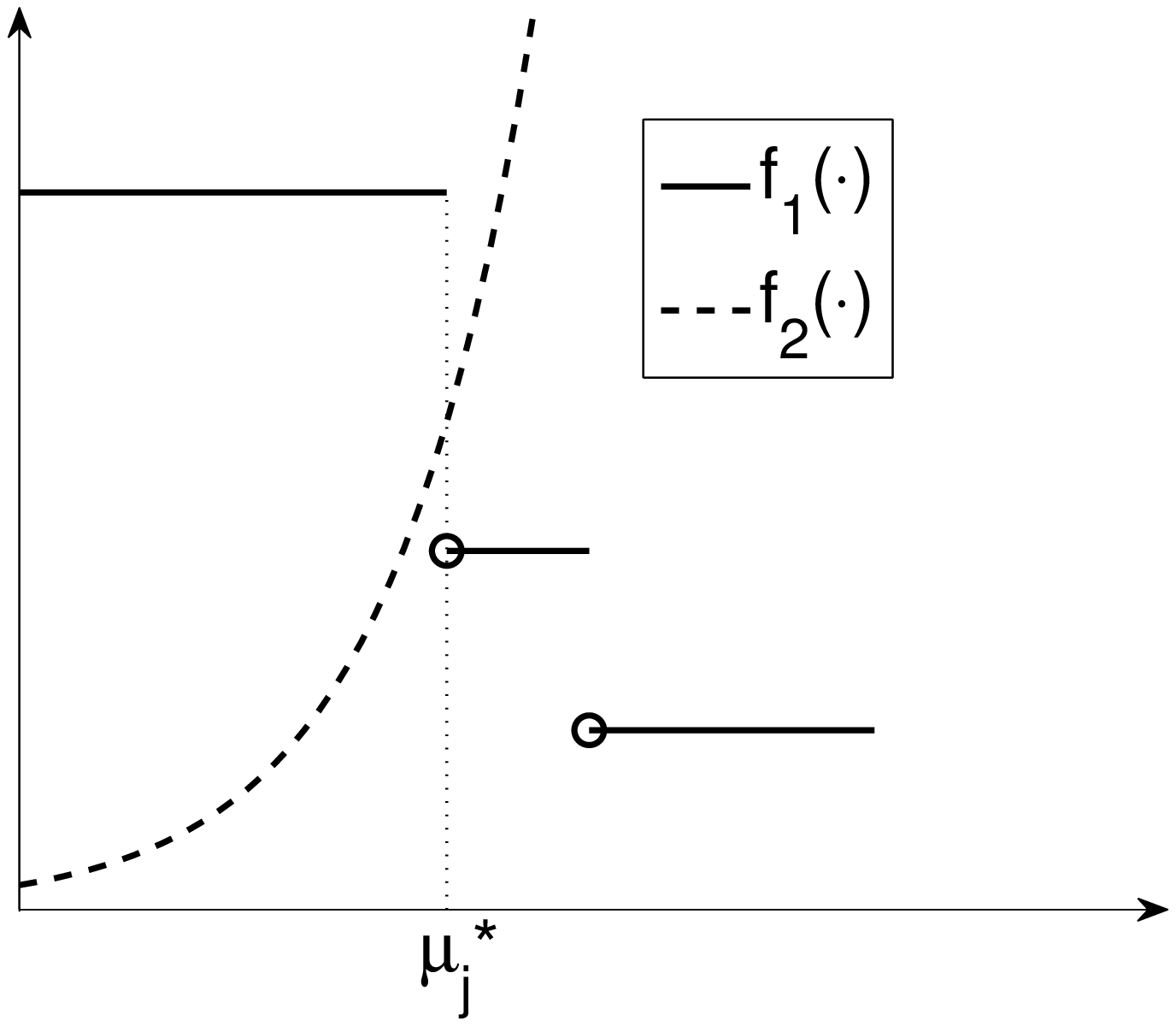}}
  \centerline{(b) no intersection}\medskip
\end{minipage}
\caption{Two cases of updating $\mu_j^*$ in dual coordinate descent }
\label{fig:f1-f2}
\end{figure}

\begin{algorithm}
\caption{Dual Coordinate Descent Method}
\label{alg:Dual Coordinated Descent}
\begin{algorithmic}
\STATE \textbf{Initialization: }Set $\mu_j=0$, $\forall j$.
Set $\nu=\log\frac{\sum_{j}e^{\mu_j-1}}{K}$. \\
\REPEAT
\FOR{each $j \in \{1,\cdots,L\}$}
\STATE
1) Update $\mu_j$ according to (\ref{mu_update}).\\
\ENDFOR \\
2) Update $\nu$ according to (\ref{nu*}). \\
\UNTIL the dual objective value converges. \\
3) Set user-BS association according to (\ref{x*}). Resolve ties if
necessary. \\
\end{algorithmic}
\end{algorithm}

The DCD method is quite intuitive. The dual variable $\mu_j$ is the
price at BS $j$, while $a_{ij}$ is the utility of the user $i$ if it
is associated with BS $j$. Each user chooses to associate with the BS that
maximizes its utility minus the price, while the BSs choose their
prices in an iterative fashion to balance their loads. Fig.~1
illustrates the price update condition that seeks $\mu_j^*$ to match
$f_1(\mu_j^*)$ and $f_2(\mu_j^*)$. Here, $f_1(\cdot)$ is a step
function. The functions $f_1(\cdot)$ and $f_2(\cdot)$ may not
intersect, but the optimal $\mu_j^*$ can always be determined
uniquely.

As mentioned earlier, a main advantage of the DCD method
is that BSs do not need to synchronize their price updates. In fact,
the order of price updates in Algorithm 1 can be arbitrary. Since each
dual update step always produces a dual objective value that is
nonincreasing, the iterative algorithm is always guaranteed to
converge.

However, it should be noted that since the dual objective
(\ref{dual_enclosed}) is not a differentiable function, coordinate
descent is not guaranteed to give a global optimum solution to the dual
optimization problem (\ref{prob:dual}), and most likely not the optimum
solution to the primal problem (because a duality gap can exist).
Nevertheless, the convergence point for DCD still gives fairly good
solutions to the original BS association problem.

The proposed dual coordinate descent method is inspired by the
development of auction algorithm \cite{bertsekas_auction} for the
one-to-one assignment problem. The BS assignment problem in this
section can be thought of as a generalization of the assignment
problem solved by the auction algorithm \cite{bertsekas_auction}
from the 1-to-1 to the $N$-to-1 case.

\subsection{Duality Gap Bound}

Although the DCD method is not guaranteed to
converge to the global optimum of the dual problem, and further,
because of the integer constraints, there may be a non-zero optimal
duality gap between the primal and the dual problems, the Lagrangian
dual analysis nevertheless gives useful upper bounds on the optimum
value of the original optimization problem.  In particular,
$g(\bm{\mu},\nu)$ is an upper
bound on $f_{\text{o}}({\mathbf{X}^*},{\mathbf{R}^*})$, and the gap
is tightest when $(\bm{\mu},\nu)$ are dual optimal.
The following result shows that this optimal duality gap can
be expressed analytically in closed form.

\newtheorem{theorem}{Proposition}
\begin{theorem}
For the BS association problem (\ref{bsa_problem}), the gap between
the objective function $f_{\text{o}}(\mathbf{X},\mathbf{R})$ obtained
from the dual coordinate descent algorithm and the global optimum utility is
bounded by $\sum_j k_j \log\left(k_j/e^{\mu_j-\nu-1}\right)$, where
$k_j$ is the number of users associated with BS $j$ and
$(\bm{\mu},\nu)$ are the values of the dual variables at convergence.
\end{theorem}
\begin{proof}
See Appendix A.
\end{proof}

Note that whenever $k_j = e^{\mu_j-\nu-1}$ for a BS $j$, as in
Fig.~1(a), the user association is close to being optimal at that BS,
as it does not contribute to the duality gap. When a BS is involved in
ties, the duality gap is minimized when $k_j$ is made as close to
$e^{\mu_j-\nu-1}$ as possible.

\section{Joint BS Association and Power Control in SISO Networks}

Thus far, we have considered the downlink BS association problem with
fixed BS transmit powers.  However, the setting of downlink power
levels is crucial for determining cell range, especially in a HetNet
where pico BSs may have a very different transmit power as compared to
macro BSs. This motivates us to investigate joint BS association and
downlink power optimization.

\begin{algorithm}[t]
\caption{Iterative BS Association and Power Control}
\label{alg:Outside Incorporation Method}
\begin{algorithmic}
\STATE \textbf{Initialization: } Set $p_j$'s to feasible values. \\
\REPEAT \STATE
1) Run DCD algorithm for fixed $p_j$'s using Algorithm 1.\\
2) Do power control for network utility maximization under fixed BS association to
	obtain new set of $p_j$'s. \\
\UNTIL convergence. \\
\end{algorithmic}
\end{algorithm}

\subsection{Iterative DCD and Power Control}

The main algorithm proposed in this section is a simple and
straightforward iteration between pricing-based BS association and
power control as shown in Algorithm 2. The idea is to run the DCD
algorithm under fixed power in order to achieve better load balancing,
and to run a power control method under fixed user association for
interference mitigation. The power optimization algorithm should also
aim to maximize the overall network utility function. One possible
implementation of such a power optimization is included in
Appendix B, where the Newton's method is used for maximizing the log utility.
As long as both the BS association and the power control
steps in the iteration aim to increase the same
objective function, the overall algorithm is guaranteed to converge
(albeit not necessarily to the global optimum, since the problem is
not convex).






Although the main idea of iterative BS association and power
optimization appears straightforward, this paper makes a key
observation that the use of utility-maximization based BS association
algorithm is crucial here.  The following simple example shows that if
instead the max-SINR association rule
is used iteratively with power control, the overall process
may not work well at all.

Consider a two-BS scenario with the initial BS assignment
as shown in Fig.~\ref{fig:bad_iteration}. If we apply power control,
BS $A$ would raise its transmit power due to the fact
that it serves a large number of users, while BS $B$ would lower its
power.
But once BS $A$ increases its power, according to the max-SINR rule,
it would attract even more users. Thus, the overall process may
exacerbate load imbalance.  This is in contrast to the pricing based
BS association, which would actually reduce the number of users served
by BS $A$ (due to the higher pricing term), hence avoiding the undesirable
phenomenon of overloading at BS $A$.

\begin{figure}[t]
\centering
\centerline{\includegraphics[width=8cm]{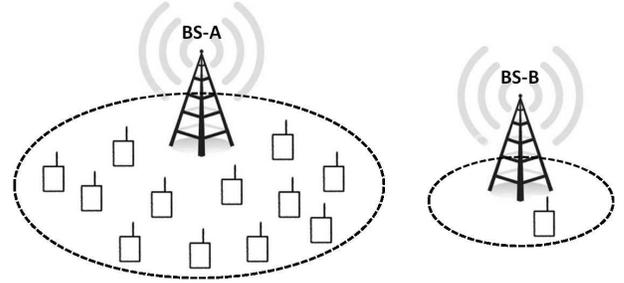}}
\caption{\label{fig:bad_iteration}An initially unbalanced BS association}
\end{figure}

\subsection{Direct Dual Optimization for Joint BS Association and Power Control}

\setcounter{equation}{19}
\begin{figure*}[b]
\hrulefill
\begin{equation}
R^{(t)}_{ij} = W\log\det\left(\mathbf{I}_{N_i\times{N_i}}+\mathbf{H}_{ij}\mathbf{v}_{ij}\mathbf{v}^H_{ij}\mathbf{H}^H_{ij}\left(\sum_{(i',j')\neq(i,j)}\mathbf{H}_{ij'}\mathbf{v}_{i'j'}\mathbf{v}^H_{i'j'}\mathbf{H}^H_{ij'}+\sigma_i^2\mathbf{I}_{N_i\times{N_i}}\right)^{-1}\right).
\label{rate_mimo}
\end{equation}
\end{figure*}
\setcounter{equation}{16}

The iterative BS association and power control method proposed in the
previous section is simple and effective. To further quantify its
performance, this section pursues an alternative direct dual
optimization approach for solving the joint BS association and power
control problem (\ref{power_problem}). The algorithm proposed in this
section is much more computationally complex than the iterative
approach of the previous section, but it serves as a benchmark for
performance comparison purpose.

The idea of direct dual optimization is to write down the Lagrangian dual of (\ref{power_problem})
\begin{multline}
\label{Lagrangian_Power}
L(\mathbf{X},\mathbf{k},\mathbf{p},\bm{\mu},\nu)= \sum_{i,j}x_{ij}\log\left(R_{ij}(\mathbf{p},k_j)\right)\\
-\sum_j{\mu_j\left(\sum_{i}x_{ij}-k_j \right)}-\nu\left(\sum_{j} k_j-K\right),
\end{multline}
which gives
\begin{equation}
g(\bm{\mu},\nu)  = \left\{
\begin{array}{cl}
\mathop\text{max}\limits_{\mathbf{X},\mathbf{k},\mathbf{p}} & L(\mathbf{X},\mathbf{k},\mathbf{p},\bm{\mu},\nu) \\
{\rm s.t.\ } & 0 \leq p_j \leq \overline{p}_j, \; \forall j \\
& \sum_{j} x_{ij}=1, \; \forall i \\
&x_{ij} \in \{0,1\}
\end{array}
\right. .
\label{dual_Power}
\end{equation}
The above maximization problem is now over both the power $\mathbf{p}$
and the BS association variables $\mathbf{X}$ and $\mathbf{k}$ under
{\it fixed} dual variables. As before, the optimal solution for $k_j$
can still be obtained analytically by (\ref{K*}), i.e.,
$k_j^*=e^{\mu_j-\nu-1}$. However, the optimization over $\mathbf{X}$
and $\mathbf{p}$ is considerably more difficult because of the
nonconvex and discrete nature of the problem. Here, we propose an approach
of iteratively optimizing $\mathbf{X}$ assuming fixed $\mathbf{p}$
using (\ref{x*}), then optimizing $\mathbf{p}$ under fixed
$\mathbf{X}$.
Clearly, the solution so obtained may not be the global optimum.
Thus, we choose to start from multiple initial points of
$(\mathbf{X},\mathbf{p})$ in order to better approach the global
optimum.

For the minimization of the dual function $g(\cdot)$, it is possible
to pursue a subgradient or dual coordinate descent approach.
The key is to recognize that the subgradient in (\ref{subgradient})
is still valid; further the optimization of $\nu$ can still be done
via (\ref{nu*}). However, the dual coordinate descent step
(\ref{mu_update}) no longer applies in a straightforward fashion.
Instead, to implement coordinate descent, a bisection method on $\mu_j$ can
be done in order to find the optimal $\mu_j^*$, while holding other
dual variables fixed.
Bisection can be carried out based on the subderivative of $g(\cdot)$
with respect to $\mu_j$, which can be calculated as
$\left(e^{\mu_j-\nu-1}-\sum_{i}x_{ij}\right)$, where $x_{ij}$ is the solution to
(\ref{dual_Power}).
Under an ideal assumption that the true optimal solution
$\left(\mathbf{X}^*,\mathbf{k}^*\right)$ can be found when evaluating
$g(\cdot)$, 
we can further deduce that this dual method has the same performance bound as in Proposition 1.
Finally as mentioned before, to ensure the near global optimum evaluation of $g(\cdot)$, multiple
random starting points need to be tried. This gives a way to find near
globally optimal solution to the overall problem.

The direct dual optimization method described above has
much higher complexity than the proposed iterative BS association and
power control method proposed in the previous section, but given
enough starting points, it can be served as a benchmark for the
proposed algorithm. The numerical simulation carried out later in the
paper indicates, however, that the simpler iterative BS association and power
control method proposed earlier already performs very close to the benchmark.


\section{Joint BS Association and Beamforming in MIMO Networks}

We now further extend the BS association problem to the case where
both the BSs and the users are equipped with multiple antennas, and
multiple users are spatially multiplexed within each cell. The use
of beamforming can significantly influence the overall effective
channel gain, and consequently the optimal BS association for each
user. Thus, the joint BS association and beamforming problem is highly
nontrivial.  Note that power control is implicitly included as part of
beamforming here.

This section first reviews the state-of-the-art in this area, then
proposes a novel approach of decoupling the overall problem into two
subproblems where the BS association and the beamformers are optimized
separately. The proposed approach has lower computational complexity;
it does not require frequent BS handover; it has comparable performance
to the best benchmark joint optimization algorithm in the literature.

\subsection{Problem Formulation and Existing Approach}

Consider a downlink MIMO cellular network with $M_j$ antennas at
BS $j$ and $N_i$ antennas at user $i$.
The channel between user $i$ and BS $j$ is denoted by matrix
$\mathbf{H}_{ij}\in\mathbb{C}^{N_i \times M_j}$. We assume one data
stream per user, and up to $M_j$ users being spatially multiplexed at
the same time. The channel is assumed to be flat-fading.  Each BS is
assumed to have a fixed total power constraint.

Because the scheduling operation, as well as transmit and receive
beamformers, are designed to adapt to the channel realizations of
each user, we can no longer claim that the proportionally fair
scheduling would result in equal time/frequency allocation among all
the users.  Instead, proportionally fair scheduling over time needs to
be included explicitly in the problem formulation.  Toward this end,
let the BS association $x_{ij}$ be fixed over time. Let
$\mathbf{v}_{ij}^{(t)}\in
\mathbb{C}^{M_j}$ be the transmit vector of BS $j$ intended for user
$i$ at time $t$. In order to maximize the network
utility defined as the log of the long-term average rates of all users, i.e.,
$\sum_{i} \log \left(R_{i}^{\rm avg}\right)$,
we can equivalently maximize a weighted rate sum over successive time slots:
\begin{subequations}
\label{prob:bsa_bf}
\begin{eqnarray}
\mathop\text{maximize}\limits_{\mathbf{X},\mathbf{V}^{(t)}} & &
\sum_{i} \omega^{(t)}_{i} \sum_j x_{ij} R^{(t)}_{ij}
\label{obj_MIMO}\\
\text{subject to}
&&\sum_{i}\|\mathbf{v}_{ij}^{(t)}\|^2 \leq \overline{p}_j, \; \forall j
\label{power_constraint_MIMO}\\
&&\sum_{j}x_{ij}=1, \; \forall i
\label{association_constraint_MIMO}\\
&&x_{ij} \in \{0,1\}, \; \forall i, \forall j
\label{bs_constraint_MIMO}
\end{eqnarray}
\end{subequations}
where $\mathbf{V}^{(t)}=\left[\mathbf{v}^{(t)}_{ij}\right]$, the weight
$\omega_{i}^{(t)}$ equals the reciprocal of each user's long-term
average rate at time $t$, and $R^{(t)}_{ij}$ is the instantaneous rate
of user $i$ at time $t$ if it is associated with BS $j$ as
expressed in (\ref{rate_mimo}) at the bottom of this page
(for ease of notation, time index  $t$ is omitted).

Note that user scheduling within each BS is implicit in the problem formulation
(\ref{prob:bsa_bf}); further in (\ref{power_constraint_MIMO}), $\overline{p}_j$
is the peak PSD constraint of BS $j$, and the constraint
(\ref{association_constraint_MIMO}) enforces the rule that each user is
associated with only one BS. Since $x_{ij}$ is not allowed to depend on $t$, for each
optimization period with a fixed set of channels, BS handovers
from time to time are not permitted.

The beamforming design problem for weighted rate-sum maximization is a
difficult nonconvex problem, even when the BS association is fixed.
Below, we briefly review a WMMSE approach for solving this problem for
fixed BS assignment, and a generalization of the WMMSE algorithm
in \cite{wmmse_luo} that accounts for BS association.

\subsubsection{Beamforming via WMMSE with Fixed BS Association}

When user-BS association is fixed in (\ref{prob:bsa_bf}), the problem
reduces to a beamforming design problem with a weighted rate sum
maximization objective. As proposed in \cite{christensen} and
\cite{wmmse}, this beamforming problem can be solved by solving an
equivalent weighted minimum mean-square error (WMMSE) problem.
We refer to \cite{wmmse} for the detailed description of the WMMSE
algorithm.

\subsubsection{WMMSE Method for BS Association}

The recent work \cite{wmmse_luo} further incorporates BS association into
the beamforming problem by imposing a penalty term to the weighted rate-sum
objective and by solving the resulting penalized WMMSE problem for each time
instant $t$.  Basically, the users are penalized for being associated with more
than one BS, and accordingly constraint (\ref{association_constraint_MIMO}) is
guaranteed in the end. However, this approach does not guarantee that the
user-BS association is fixed over time. Consequently, as weights $\omega_i$ are
updated over time, user association and user scheduling can both change. This
results in rapid BS handovers, which are not desirable in practice.


Further, the WMMSE-based BS
association method as proposed in \cite{wmmse_luo} has
high computational complexity, because the WMMSE update needs to be done
between every single BS-user pair in the entire network. Also, the
performance and convergence speed of the algorithm depend heavily on the
parameter of the penalty term, which can only be set heuristically. Nevertheless, the
method of \cite{wmmse_luo} provides a useful benchmark for our proposed
algorithm below.

\subsection{Proposed Two-Stage BS Association and Beamforming}

This paper formulates the joint BS association and beamforming problem in
recognition of the fact that BS association typically takes place at a much
larger time scale and should only adapt to the slow-fading channel
characteristics, while beamforming and scheduling can take place in faster time
scale.  Thus, instead of jointly optimizing BS association and beamforming at
each time slot, it is more sensible to decouple them in two stages. The first
stage solves the BS association problem, while the second stage solves the
beamforming problem assuming fixed BS association. The proposed two-stage
algorithm is described below:

\subsubsection{BS Association Stage}

The idea is to determine BS association in the first stage based on
an \emph{estimate} of channel quality. For BS association purposes, we
rely on a simple SISO representation of the MIMO channel,
and apply the joint coordinate descent and power control
algorithm presented in the previous section to determine the BS
association for each user.

The SISO representation for the MIMO channel is based on the fact that from a
degree of freedom point of view, $M_j$ antennas at the BS provide $M_j$ spatial
multiplex gain. 
Thus, we can think of a MIMO system with $M_j$ antennas over
bandwidth $W$ as equivalently a SISO system with bandwidth $M_j W$.
More precisely, let
$|h_{ij}|$ be the average channel magnitude between BS $j$ and user $i$
(modeling the distance-dependent attenuation and
shadowing). We estimate each user's SINR according to (\ref{sinr}),
while accounting for the multiple $M_j$ antennas at the BS by
redefining parameter $a_{ij}$ as
\setcounter{equation}{20}
\begin{equation}
\label{new_a_1}
a_{ij} =
\log\left(M_jW\log\left(1+\frac{\text{SINR}_{ij}}{\Gamma}\right)\right).
\end{equation}
The joint BS association and power control algorithm can now be applied to
determine the BS association.  We remark here that only the BS association is
of interest at this stage. The optimized power $p_j$ serves to assist the BS
association and scheduling decisions, but is further optimized in the
next stage.

\subsubsection{Scheduling and Beamforming Stage}

After the BS association is determined, the overall problem now reduces to the
beamforming vector design problem, which can be solved using the WMMSE
algorithm. Our contribution in algorithm design in this stage is to point out
that one can further lower the computational complexity of WMMSE by eliminate candidate
users that are unlikely to be scheduled.

In the conventional WMMSE algorithm, all potential users within a cell can
have their beamforming vectors updated in each step.  However, because each BS
$j$ can spatially multiplex at most $M_j$ users, to reduce the computational
complexity, we may choose a subset of users who are most likely to be served to
take part in the WMMSE algorithm.  The simplest way to do this is to choose the
users according to the estimated weighted rate $\omega_{i} x_{ij} \tilde{R}_{ij}$,
where $\tilde{R}_{ij}$ is calculated by the SISO model (\ref{rate_BS}) scaled
by $M_j$ according to the resulting $\left(\mathbf{X},\mathbf{p}\right)$ after
stage one.  More sophisticated scheduling can also take channel directions into
account.  The number of potential users chosen by the WMMSE scheduler in cell
$j$ is a parameter, called $S_j$ in this paper, which should be greater than
$M_j$. 
A complete description of the two-stage method is stated in Algorithm 3.

\begin{algorithm}[t]
\caption{Two-Stage Joint BS Association and WMMSE Beamforming}
\label{alg:Two-stage Method}
\begin{algorithmic}
\STATE \textbf{Initialization: }Choose $S_j\geq M_j$, $\forall j$.\\
1) Run Algorithm 2, the joint BS association and power control (with $a_{ij}$
calculated by (\ref{new_a_1})) until convergence.  Let the result of the
optimization be $(\mathbf{X},\mathbf{p})$.  Associate users to BSs according
to $\mathbf{X}$. Compute $\tilde{R}_{ij}$ according to
$(\mathbf{X},\mathbf{p})$ using the SISO model (\ref{rate_BS}) scaled by
$M_j$.\\
\REPEAT
\STATE
2) Choose $S_j$ potential users among the users associated with BS $j$
according to $\omega_{i} x_{ij} \tilde{R}_{ij}$, $\forall j$.\\
3) Run the WMMSE algorithm \cite{wmmse} for the chosen users in
each cell to get the transmit beamformers and the resulting rate $R^{(t)}_{ij}$.\\
4) Update the average rate for each user $R^{\text{avg}}_{i}$ based on
$R^{(t)}_{ij}$; set $\omega_{i} =  1/R^{\text{avg}}_{i}$, $\forall i$.
\\ \UNTIL $R^{\text{avg}}_{i}$ converges, $\forall i$. \\
\end{algorithmic}
\end{algorithm}

\subsection{Complexity Analysis}
This subsection briefly analyzes the computational complexity saving of the
proposed algorithm as compared of the joint BS association and beamforming
algorithm of \cite{wmmse_luo}.  For simplicity, we assume that the number of
antennas at all the BSs are the same and the number of antennas at all the
users are the same, i.e., $M_j=M$ for all $j$'s, and $N_i=N$ for all $i$'s.
Under fixed BS association, the conventional WMMSE algorithm has a complexity
of ${O}\left(K^2MN^2+K^2M^2N+KM^3+KN^3\right)$ per each beamforming step,
where $K$ is the number of users in the entire network.

For the joint BS association and WMMSE method of \cite{wmmse_luo}, since the
WMMSE update of each user needs to be done with respect to all $L$ BSs in the
network, parameter $K$ in the WMMSE complexity formula for the algorithm of
\cite{wmmse_luo} needs to be increased by a factor of $L$, resulting in a
complexity of ${O}\left(L^2K^2MN^2+L^2K^2M^2N+LKM^3+LKN^3\right)$.

By contrast, in the proposed two-stage algorithm, only $S = \sum_jS_j$ users
are considered, and they are already associated with their respective BSs.
Consequently, the complexity per each WMMSE iteration is reduced to
of ${O}\left(S^2MN^2+S^2M^2N+SM^3+SN^3\right)$. Since $S \ll K \ll LK$, this is
significant complexity saving.  In the above calculation, we ignore the
complexity of the first stage, which is typically very fast. In addition, we
do not account for the number of iterations in the WMMSE algorithm. However,
the number of WMMSE iterations is typically smaller for the proposed algorithm
than for the WMMSE algorithm of \cite{wmmse_luo}
since fewer users are involved. Overall, the proposed two-stage algorithm is
much faster than the WMMSE algorithm of \cite{wmmse_luo}. The simulation results of next section show that it performs
almost as well.


\section{Simulation Results}

\begin{table}[!t]
\footnotesize
\renewcommand{\arraystretch}{1.3}
\label{simulation_parameters}
\caption{Simulation Parameters}
\centering
\begin{tabular}{|c|c|}
\hline
Channel Bandwidth & 10 MHz \\
\hline
Frequency Reuse Factor & 1 \\
\hline
Duplex Mode & TDD \\
\hline
Macro BS Max PSD & -27 dBm/Hz \\
\hline
Pico BS Max PSD & -47 dBm/Hz \\
\hline
Antenna Gain & 15 dBi  \\
\hline
SNR Gap & 0 dB \\
\hline
Background Noise PSD & -169 dBm/Hz \\
\hline
Distance-dependent Attenuation & $128.1+37.6\log_{10}(d)$, $d$ is in km \\
\hline
Shadowing & Log normal as $\mathcal{N}(0,\sigma^2)$, $\sigma=8$dB \\
\hline
\end{tabular}
\end{table}

\subsection{BS Association Under Fixed Powers}

We first simulate the BS association algorithms with fixed powers
in a downlink SISO network with a 7-cell wrap around topology, with
one macro-BS and three pico-BSs per cell, and with 30 users per cell.
The channel modeling parameters are as defined in Table I.
The transmit PSD level is fixed at the maximum value for each BS.
Fig.~\ref{fixed_duality} compares the convergence behavior of the dual
coordinate descent (Algorithm 1) with that of the adaptive
subgradient method. Here each iteration refers to either a single
update of $\mu_j$ in the DCD method or a subgradient update of all
$\mu_j$'s. We see that the DCD method converges to within $10^{-1}$ of
the optimum with only two rounds of iterations per BS (i.e. 56
iterations), while the convergence of subgradient method is very
sensitive to its parameters.  Here, we set $\rho=1.2$, $\beta=0.9$,
and $\delta=0.002$ in the adaptive subgradient method
\cite{andrews, bertsekas_convex} and see that different settings of $\delta_1$
and $\gamma_k$ can result in very different convergence behaviors.
Note that in Fig. \ref{fixed_duality} the DCD method does not converge to
the optimum. This is due to the fact that it is possible for
coordinate descent to get stuck in a suboptimal point. This gap is
quite small in this simulation, however.

Fig.~\ref{fixed_cdf} shows the cumulative distribution function (CDF) of data
rates after 56 iterations for the various BS assignment algorithms. We see
that both the subgradient method and the DCD method offer substantial rate
improvement to low-rate users as compared to the max-SINR BS assignment rule.
For instance, the 50th-percentile rate is increased by about 33\%, which is a
consequence of off-loading traffic from the macro BSs to the pico BSs.  The
performance of the subgradient method is again parameter dependent.

Table II shows that the numerical utility\footnote{The numerical value
of the utility is computed as sum of log of user rates,
where rates are in Mbps.} achieved by DCD and two
of the subgradient methods are almost identical, while subgradient-2
and the max-SINR method produce quite inferior results.
This is consistent with the earlier convergence plot (Fig.~\ref{fixed_duality})
and the CDF plot (Fig.~\ref{fixed_cdf}).
In addition, the duality-gap bound calculated according to Proposition 1
for this example is about 0.45. This shows that the performance of the DCD
algorithm is already very close to the global optimum. Finally,
Fig.~\ref{fixed_num_user} displays the percentages of macro/pico users for
various BS association methods. It shows that with the max-SINR BS association
and subgradient-2, too many users are associated with the macro BS, while the
DCD algorithm is able to achieve more balanced load by off-loading the users to
pico BSs .


\begin{figure}[t]
\begin{minipage}[b]{1.0\linewidth}
\centering
\centerline{\includegraphics[width=8cm]{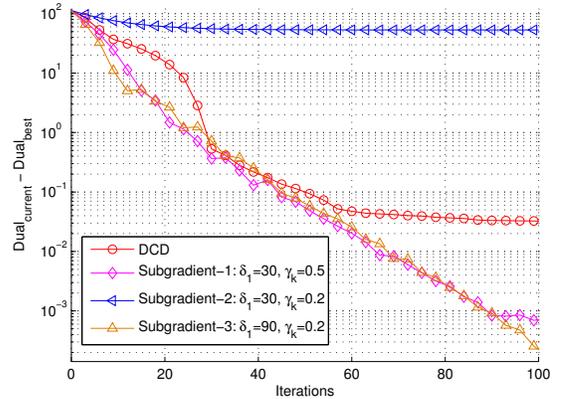}}
\caption{Convergence behaviors of dual coordinate descent and subgradient algorithms}
\label{fixed_duality}
\end{minipage}
\end{figure}
\begin{figure}[t]
\begin{minipage}[b]{1.0\linewidth}
\centering
\centerline{\includegraphics[width=8cm]{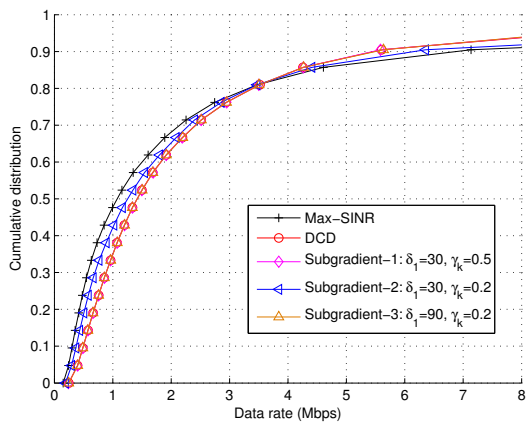}}
\caption{CDF of user rates for various BS association methods after 56 iterations}
\label{fixed_cdf}
\end{minipage}
\end{figure}

\begin{figure}[!t]
\begin{minipage}[b]{1.0\linewidth}
\centering
\centerline{\includegraphics[width=8cm]{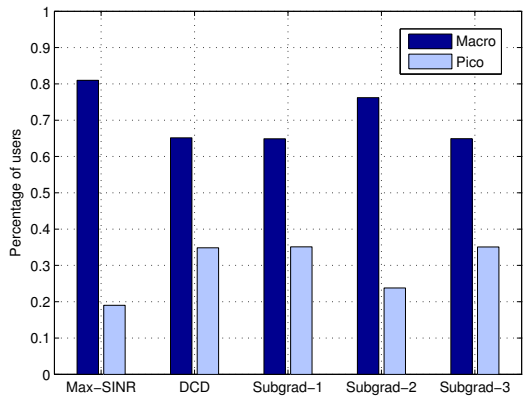}}
\caption{The percentages of macro/pico users for various BS association methods after 56 iterations}
\label{fixed_num_user}
\end{minipage}
\end{figure}

\begin{table*}[!t]
\small
\label{tab_utility_BSA_flat}
\caption{Utility values for various BS association
methods after 56 iterations}
\centering
\begin{tabular}{p{1.5cm}<{\centering}|p{2.6cm}<{\centering}|p{2.6cm}<{\centering}|
p{2.6cm}<{\centering}|p{2.6cm}<{\centering}|p{2.6cm}<{\centering}}
\hline
 & Max-SINR & DCD & Subgradient-1 & Subgradient-2 & Subgradient-3\\
\hline
Utility & 52.86 & 97.63 & 97.58 & 75.04 & 97.66 \\
\hline
\end{tabular}
\end{table*}

\subsection{Joint BS Association and Power Control}

\begin{figure}[t]
\begin{minipage}[b]{1.0\linewidth}
\centering
\centerline{\includegraphics[width=8cm]{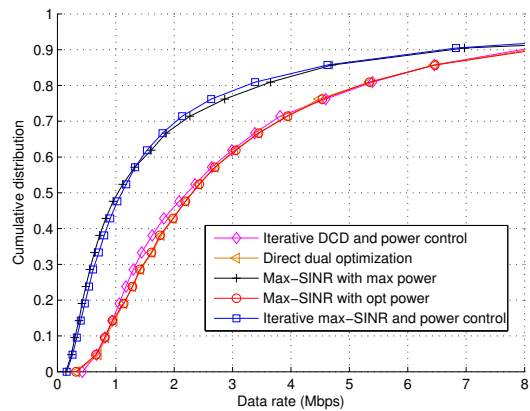}}
\caption{CDF of user rates for various joint BS association and power control methods}
\label{pc_cdf}
\end{minipage}
\end{figure}
\begin{figure}[t]
\begin{minipage}[b]{1.0\linewidth}
\centering
\centerline{\includegraphics[width=8cm]{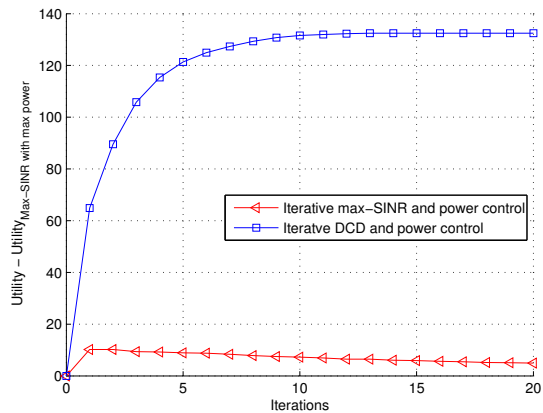}}
\caption{The convergence of the iteration with power control for DCD and max-SINR}
\label{pc_convergence}
\end{minipage}
\end{figure}

\begin{figure}[t]
\centering
\centerline{\includegraphics[width=9.5cm]{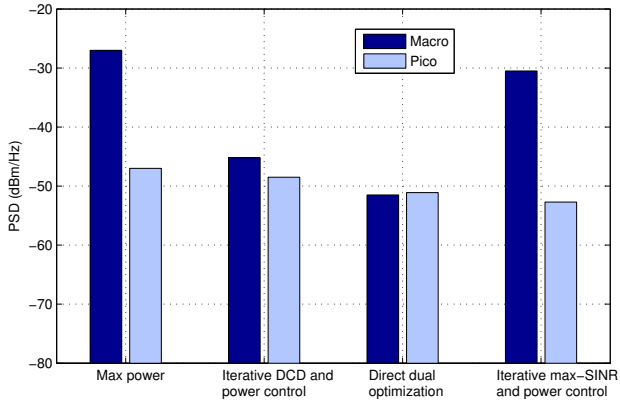}}
\caption{Optimized powers for various joint BS association and power control methods}
\label{pc_power}
\end{figure}

\begin{figure}[t]
\centering
\centerline{\includegraphics[width=9.5cm]{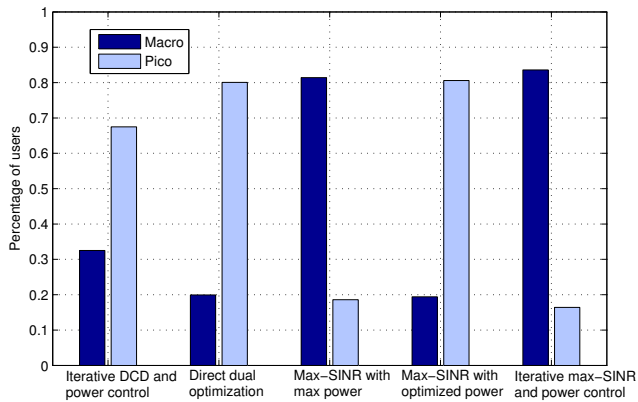}}
\caption{The percentages of maco/pico users for various joint BS association and power control methods}
\label{pc_num_user}
\end{figure}

\begin{table*}[t]
\small
\renewcommand{\arraystretch}{1.3}
\label{flat_BSA_utility}
\caption{Utility values for various joint BS association and power control methods}
\centering
\begin{tabular}{p{1.5cm}<{\centering}|p{2.6cm}<{\centering}|p{2.6cm}<{\centering}|p{2.6cm}<{\centering}
|p{2.6cm}<{\centering}|p{2.6cm}<{\centering}}
\hline
& Iterative DCD and power control & Direct dual optimization & Max-SINR with max power & Max-SINR with optimized power & Iterative max-SINR and power control \\
\hline
Utility & 186.29 & 194.41 & 52.86 & 193.01 & 56.09\\
\hline
\end{tabular}
\end{table*}

This section considers the same network topology, but with downlink power control
implemented in addition. We use Newton's method
for power control for utility maximization.
Note that since the network utility maximization problem is nonconvex, only the convergence to local optimum is expected.
For the implementation of the
direct dual optimization approach, we choose 10 random starting points.

In Fig.~\ref{pc_cdf}, we observe a significant difference between
max-SINR BS association and DCD-based BS association when they are implemented
iteratively with power control. Further, Fig.~\ref{pc_convergence} shows that
the iteration between DCD and power control gives incremental improvement in
utility, while in the max-SINR case utility actually decreases
after the second iteration. These two plots validate the earlier analysis
showing that the max-SINR association does not address the load balancing issue
effectively and that the use of utility-maximization-based BS association is
crucial when implemented with power control.

As can be seen in Fig.~\ref{pc_cdf} and Table III, the direct dual optimization
approach is able to provide the best performance among all the methods, but at
the cost of very high complexity. In the simulation, we observe that during the
updating of one single dual variable, direct dual optimization needs to call
the power control algorithm approximately 1000 times, while the iterative DCD
and power control method only needs to run the power control method once in
each iteration.

For comparison purpose, we also implement the max-SINR BS association under
the powers optimized by the duality-based approach.  Now, max-SINR performs
well as seen in the Fig.~\ref{pc_cdf} and Table III.  This shows that the
problem with the max-SINR algorithm is that it is unable to induce the correct
power setting, in contrast to the DCD scheme.

In Fig.~\ref{pc_power}, we show the PSD levels produced by the various methods.
It is observed that the methods with better performance are able to suppress
the overly high transmit power by the macro BSs. Further,
Fig.~\ref{pc_num_user} shows the percentages of users associated with the macro
and pico BSs resulting from
various methods. Methods with better performance tend to have higher
percentages of pico users, which illustrates the benefit of off-loading traffic
from macro BSs to pico BSs.
Combining results from Fig.~\ref{pc_power} and Fig.~\ref{pc_num_user}, we
conclude that a combination of suppressing macro BS power for interference
mitigation and off-loading to pico BSs for load balancing is the key to
obtaining overall good system performance.


\subsection{Joint BS Association and Beamforming}

\begin{figure}[!t]
\begin{minipage}[b]{1.0\linewidth}
  \centering
	  \centerline{\includegraphics[width=8cm]{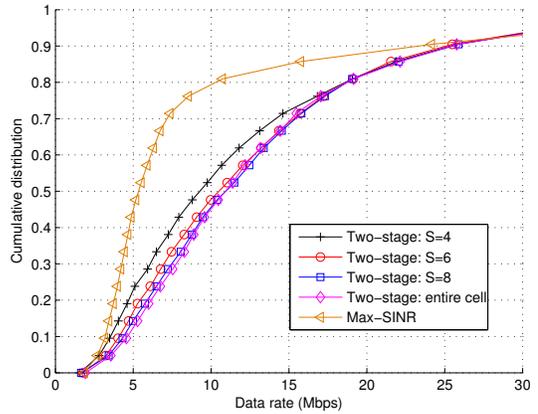}}
  \caption{CDF of user rates for joint BS association and beamforming:
  two-stage method vs. max-SINR for a network with 7 macro BSs and 21 pico BSs}
  \label{mimo_cdf}
\end{minipage}
\end{figure}
\begin{figure}
\begin{minipage}[b]{1.0\linewidth}
  \centering
	  \centerline{\includegraphics[width=8cm]{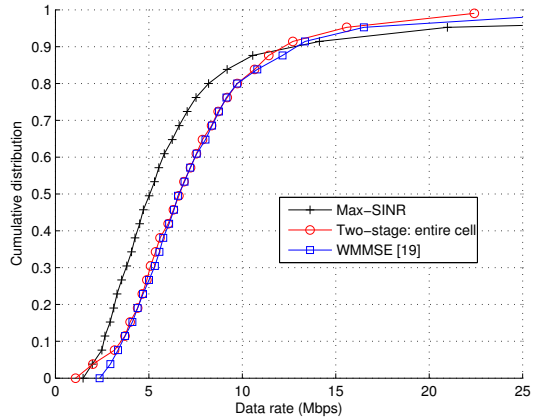}}
  \caption{CDF of user rates for joint BS association and beamforming:
  two-stage method vs. WMMSE [19] for a network with 3 macro BSs and 4 pico BSs}
  \label{mini_cdf}
\end{minipage}
%
\end{figure}

\begin{table*}[!t]
\small
\renewcommand{\arraystretch}{1.3}
\label{flat_BSA_utility}
\caption{Utility values for various joint BS association and
beamforming methods}
\centering
\begin{tabular}{p{1.5cm}<{\centering}|p{2.6cm}<{\centering}|p{2.6cm}<{\centering}|p{2.6cm}<{\centering}
|p{2.8cm}<{\centering}|p{2.6cm}<{\centering}}
\hline
 & Two-stage: $S_j=4$ & Two-stage: $S_j=6$ & Two-stage: $S_j=8$ & Two-stage: entire cell & Max-SINR \\
\hline
Utility & 468.2 & 488.05 & 495.16 & 498.62 & 392.96 \\
\hline
\end{tabular}
\end{table*}


Consider again the same network topology, but for the MIMO case with 4
antennas at each of the macro and pico BSs and 2 antennas at each user.
The two-stage BS association and WMMSE algorithm is compared with the
max-SINR BS association under maximum power plus per-cell WMMSE, in
Fig.~\ref{mimo_cdf} and in Table IV. The number of candidate users (the $S_j$
parameter) in the two-stage method is chosen to be 4, 6 and 8.
It is observed that the two-stage method can substantially
improve the max-SINR BS association:
the 50th-percentile rate is almost doubled when $S=8$. We
also observe that the performance of the two-stage method improves
with larger $S$, but the improvement beyond $S=8$ is marginal for this
case with 4 transmit antennas.

We also wish to compare the two-stage method with the joint BS association and
WMMSE method proposed in \cite{wmmse_luo}. Because this method involves
implementing the WMMSE algorithm across the entire network, its complexity is
very high. In fact, running such an algorithm across a 7-cell network (with 28
BSs) is already impractical.  Instead, Fig.~\ref{mini_cdf} compares the two
algorithms in a smaller network with 3 macro BSs, 4 pico BSs, and with 105 user
terminals.  We observe in our simulation that the utility gains by the
two-stage method and the WMMSE method of \cite{wmmse_luo} are 17.82 and 23.05
respectively as compared to the max-SINR scheme.  Although
the WMMSE method of \cite{wmmse_luo} produces overall better network utility,
we observe from Fig.~\ref{mini_cdf} that the
majority of users do not see much performance
difference between the two.  In addition, we observe in the simulation that the
joint BS association and WMMSE method of \cite{wmmse_luo} causes approximately 24 BS association
switchings on average for each beamforming update. About 1/4 of the users
are involved in BS handover in each time slot, which is not very practical.
In contrast, BS association is completely fixed in the two-stage method,
which is a clear advantage.

\section{Conclusion}

This paper considers pricing-based BS association schemes for
heterogeneous networks and proposes a distributed price update
strategy based on a coordinate descent algorithm in the dual domain.
The proposed BS association scheme can be seamlessly incorporated with
power control and beamforming.
In each of these cases, because BS assignment must be determined at a
relatively larger time scale, we propose to implement BS association
with respect to the expected average channel gains.  The overall main
insight of this paper is that load balancing is crucial in
heterogeneous networks. Instead of assigning BSs according to SINR, a
utility maximization and pricing strategy can be adopted in order to
achieve balanced loads across the network, and pricing update can be
done in a distributed fashion efficiently.

\appendices

\section{Proof of Proposition 1}

Let $(\bm{\mu},\nu)$ be the optimized dual variables at convergence
of the DCD algorithm. Let $(\mathbf{X},\mathbf{k})$ be the primal
solution recovered from the dual variable $(\bm{\mu},\nu)$ using
(\ref{x*}) with tie-breaking if necessary, and subsequently setting
$k_j = \sum_i x_{ij}$ as the number of users associated with each BS.
Let $\mathbf{R}$ be the corresponding user rates calculated by
(\ref{rate_BS}).  We have:
\allowdisplaybreaks
\begin{subequations}
\begin{eqnarray}
f_{\text{o}}(\mathbf{X},\mathbf{R})
&=&\sum_{i,j}a_{ij}x_{ij}-\sum_{j}k_j\log(k_j) \\
&=&\sum_{i,j}a_{ij}x_{ij}-\sum_{j}k_j\log\left(e^{\mu_j-\nu-1}\right)\notag\\
&&-\sum_j k_j\log\left(\frac{k_j}{e^{\mu_j-\nu-1}}\right)\\\
&=&\sum_i(a_{ij}-{\mu}_j){x}_{ij}+\sum_{j}{k}_j\notag\\
&& +\sum_{j}{\nu}{{k}_j}-\sum_j{k}_j\log\left(\frac{{k}_j}{e^{{\mu}_j-{\nu}-1}}\right)\\
&=&\sum_i\max_{j}(a_{ij}-{\mu}_j)+K+{\nu}{K}\notag\\
&&-\sum_j{k}_j\log\left(\frac{{k}_j}{e^{{\mu}_j-{\nu}-1}}\right)
\label{duality_gap1}\\
&=&\sum_i\max_{j}(a_{ij}-{\mu}_j)+\sum_{j}e^{{\mu}_j-{\nu}-1}\notag\\
&&+{\nu}{K}-\sum_j{k}_j\log\left(\frac{{k}_j}{e^{{\mu}_j-{\nu}-1}}\right)
\label{duality_gap2}\\
&=&g({\bm{\mu}},{\mathbf{\nu}})-\sum_j{k}_j \log\left(\frac{{k}_j}{e^{{\mu}_j-{\nu}-1}}\right)
\label{err_bound}
\end{eqnarray}
\end{subequations}
where the optimality condition on $x_{ij}$, (\ref{x*}), is used in
deriving (\ref{duality_gap1}), and the optimality condition on $\nu$,
(\ref{nu*}), is used in deriving (\ref{duality_gap2}).

Now, let $(\mathbf{X}^*,\mathbf{k}^*)$ be the optimal solution to
problem (\ref{bsa_problem}), and let $\mathbf{R}^*$ be the resulting
user rates. By weak duality, it always holds that $g({\bm{\mu}},{\nu})
\geq f_{\text{o}}(\mathbf{X}^*,\mathbf{R}^*)$. Combining this result
with (\ref{err_bound}), we prove the claim
\begin{equation}
f_{\text{o}}({\mathbf{X}},{\mathbf{R}}) \geq
f_{\text{o}}(\mathbf{X}^*,\mathbf{R}^*)-\sum_j {k}_j
\log\left(\frac{{k}_j}{e^{{\mu}_j-{\nu}-1}}\right).
\end{equation}

\section{Newton's method for downlink power control}
In this appendix, we describe a Newton's method for solving
the power optimization problem for maximizing the network
log utility. Assuming fixed user association $\mathbf{X}$ (and accordingly
$k_j=\sum_ix_{ij}$), the optimization problem is:
\begin{subequations}
\label{prob:power_prob}
\begin{eqnarray}
\mathop\text{maximize}\limits_{\mathbf{p}} & &
\sum_{i,j}x_{ij}\log\left(R_{ij}(\mathbf{p},k_j)\right)
\label{obj_power_prob} \\
\text{subject to}&&0 \leq p_j \leq \overline{p}_j, \; \forall j
\end{eqnarray}
\end{subequations}
Let $f_{\text{power}}(\mathbf{p})$ denote
the objective function above. Introduce parameter $r_{ij}$ as
\begin{equation}
r_{ij}=\log\left(1+\frac{\text{SINR}_{ij}}{\Gamma}\right).
\end{equation}
We can write the first-order and the second-order partial derivatives of $f_{\text{power}}(\mathbf{p})$ with respect to $p_j$ as:
\begin{multline}
\label{first_derivative}
\frac{ \partial f_{\text{power}} }{\partial p_j}
= \sum_{i} \frac{\text{SINR}_{ij}}{r_{ij}(\Gamma+\text{SINR}_{ij})}\frac{x_{ij}}{p_j}\\
- \sum_{i}\sum_{j'\neq j}\frac{|h_{ij}|^2\text{SINR}_{ij'}^2}{|h_{ij'}|^2r_{ij'}(\Gamma+\text{SINR}_{ij'})}\frac{x_{ij'}}{p_{j'}}
\end{multline}
and
\begin{multline}
\label{second_derivative}
\frac{ \partial^2 f_{\text{power}}} {\partial p_j^2}
= -\sum_{i} \left(\frac{1}{r_{ij}^2}+\frac{1}{r_{ij}}\right)\frac{\text{SINR}_{ij}^2}{(\Gamma+\text{SINR}_{ij})^2}\frac{x_{ij}}{p_j^2}\\
+
\sum_{i}\sum_{j'\neq j}
\frac{|h_{ij}|^4\text{SINR}_{ij'}^3\left(2r_{ij'}\Gamma+
\text{SINR}_{ij'}(r_{ij'}-1)\right)}{|h_{ij'}|^4r_{ij'}^2(\Gamma+\text{SINR}_{ij'})^2}\frac{x_{ij'}}{p_{j'}^2}.
\end{multline}
Following the heuristic in \cite{power}, we only use the diagonal
entries of Hessian matrix in Newton's method in order to reduce the
computational complexity of inverting the Hessian. In this case, the Newton step becomes
$\Delta p_j = -\frac{ \partial f_{\text{power}} }{\partial p_j}/
\frac{ \partial^2 f_{\text{power}}} {\partial p_j^2}$. To ensure
an  incremental updating direction, we further modify the Newton step as
\begin{equation}
\label{newton_step}
\Delta p_j = \frac{ \partial f_{\text{power}} }{\partial p_j} \left/
\left|\frac{ \partial^2 f_{\text{power}}} {\partial p_j^2} \right| \right. .
\end{equation}
The overall algorithm updates all $p_j$'s through
\begin{equation}
\label{newton_update}
p_j^{(t+1)} = \left[p_j^{(t)} + \alpha_{\text{nt}}\Delta
p_j\right]_0^{\overline{p}_j},
\end{equation}
where $\alpha_{\text{nt}}$ is the step size, which can be determined
by backtracking line search \cite{boyd_convex}.

\bibliographystyle{IEEEbib}
\bibliography{IEEEabrv,references}

\begin{IEEEbiography}
[{\includegraphics[width=1in,height=1.25in,clip,keepaspectratio]{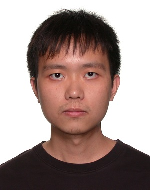}}]
{Kaiming Shen}
(S'13) received the B.Eng. degree in Information Security and the B.S. degree in Mathematics
from Shanghai Jiao Tong University, Shanghai, China in 2011 and the M.A.Sc. degree in Electrical
and Computer Engineering from the University of Toronto, Ontario, Canada in 2013. His
research interests include wireless communications and mathematical optimization.
\end{IEEEbiography}

\begin{IEEEbiography}
[{\includegraphics[width=1in,height=1.25in,clip,keepaspectratio]{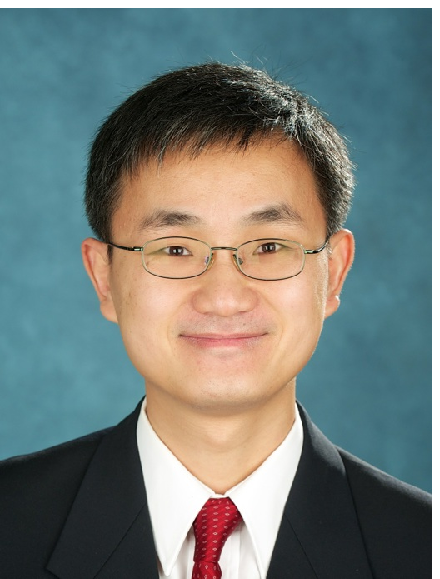}}]
{Wei Yu}
(S'97-M'02-SM'08-F'14) received the B.A.Sc. degree in Computer Engineering and Mathematics from the University of Waterloo, Waterloo, Ontario, Canada in 1997 and M.S. and Ph.D. degrees in Electrical Engineering from Stanford University, Stanford, CA, in 1998 and 2002, respectively. Since 2002, he has been with the Electrical and Computer Engineering Department at the University of Toronto, Toronto, Ontario, Canada, where he is now Professor and holds a Canada Research Chair (Tier 1) in Information Theory and Wireless Communications. His main research interests include information theory, optimization, wireless communications and broadband access networks.

Prof. Wei Yu served as an Associate Editor for IEEE Transactions on Information Theory (2010-2013), as an Editor for IEEE Transactions on Communications (2009-2011), as an Editor for IEEE Transactions on Wireless Communications (2004-2007), and as a Guest Editor for a number of special issues for the IEEE Journal on Selected Areas in Communications and the EURASIP Journal on Applied Signal Processing. He was a Technical Program Committee (TPC) co-chair of the Communication Theory Symposium at the IEEE International Conference on Communications (ICC) in 2012, and a TPC co-chair of the IEEE Communication Theory Workshop in 2014. He was a member of the Signal Processing for Communications and Networking Technical Committee of the IEEE Signal Processing Society (2008-2013). Prof. Wei Yu received an IEEE ICC Best Paper Award in 2013, an IEEE Signal Processing Society Best Paper Award in 2008, the McCharles Prize for Early Career Research Distinction in 2008, the Early Career Teaching Award from the Faculty of Applied Science and Engineering, University of Toronto in 2007, and an Early Researcher Award from Ontario in 2006.

Prof. Wei Yu is a Fellow of IEEE. He is a registered Professional Engineer in Ontario.
\end{IEEEbiography}

\end{document}